%% file: main.tex
\documentclass[11pt]{article}

\usepackage{fullpage}

\usepackage[utf8]{inputenc}
\usepackage[english]{babel}
\usepackage[compact]{titlesec}

\usepackage{latexsym, amssymb, bbm}
\usepackage{amsmath}
\usepackage{amsthm}
\usepackage{verbatim}
\usepackage{enumitem}
\usepackage{mathtools}

\usepackage{graphicx}
\usepackage{nicefrac,xspace}
\usepackage{authblk}

\usepackage[dvipsnames]{xcolor}
\usepackage[breaklinks]{hyperref}
\hypersetup{colorlinks=true,%
            citebordercolor={.6 .6 .6},linkbordercolor={.6 .6 .6},%
citecolor=blue,urlcolor=black,linkcolor=blue,pagecolor=black}
\usepackage[nameinlink]{cleveref}

\usepackage{algorithm}
\usepackage[noend]{algpseudocode}

\newtheorem{theorem}{Theorem}[section]
\newtheorem{lemma}[theorem]{Lemma}
\newtheorem{proposition}[theorem]{Proposition}

\theoremstyle{definition}
\newtheorem{definition}{Definition}
\newtheorem{remark}{Remark}

\theoremstyle{remark}

\newcommand{\cut}[1]{}

\input{def}

\title{Dimension-Free Bounds for Chasing Convex Functions\thanks{This
    research was supported in part by NSF award CCF-1907820.}}
\author[1]{C.J.~Argue} 
\author[1]{Anupam Gupta}
\author[2]{Guru Guruganesh}
\affil[1]{Carnegie Mellon University}
\affil[2]{Google Research}

\newcommand{\gr}{\nabla}
\newcommand{\grad}{\nabla}

\newcommand{\cost}{\mathsf{cost}}
\renewcommand{\k}{\kappa}
\renewcommand{\th}{\theta}
\renewcommand{\b}{\beta}

\newcommand{\MtoM}{\ensuremath{\mathsf{M2M}}\xspace}
\newcommand{\OBD}{\ensuremath{\mathsf{OBD}}\xspace}
\newcommand{\COBD}{\ensuremath{\mathsf{COBD}}\xspace}

\newcommand{\DTPhi}{\widetilde{\Delta}_t \Phi} %

\usepackage{nicefrac}
\usepackage{cleveref}

\usepackage{thm-restate,color,xcolor,xspace}

\usepackage[colorinlistoftodos,textsize=tiny,textwidth=2cm,color=red!25!white,obeyFinal]{todonotes}

\begin{document}

\maketitle

\begin{abstract}
  We consider the problem of chasing convex functions, where functions
  arrive over time. The player takes actions after seeing the
  function, and the goal is to achieve a small function cost for these
  actions, as well as a small cost for moving between actions. While
  the general problem requires a polynomial dependence on the
  dimension, we show how to get dimension-independent bounds for
  well-behaved functions. In particular, we consider the case where
  the convex functions are $\k$-well-conditioned, and give an
  algorithm that achieves an $O(\sqrt \k)$-competitiveness. Moreover,
  when the functions are supported on $k$-dimensional affine
  subspaces---e.g., when the function are the indicators of some
  affine subspaces---we get $O(\min(k, \sqrt{k \log T}))$-competitive 
  algorithms for request sequences of length $T$. We also
  show some lower bounds, that well-conditioned functions require  
  $\Omega(\k^{1/3})$-competitiveness, and that $k$-dimensional functions
  require $\Omega(\sqrt{k})$-competitiveness.
\end{abstract}

\input{intro}

\input{algos}

\input{low-dim}

{\small
\bibliographystyle{alpha}
\bibliography{well-conditioned}
}

\appendix

\input{lower-bounds}

\input{general-norm}

\end{document}

%% file: def.tex
\def\E{\mathbb{E}}

\def\R{\mathbb{R}}

\def\a{\alpha}
\def\b{\beta}
\def\g{\gamma}

\def\e{\varepsilon}

\def\th{\theta}

\def\k{\kappa}

\def\la{\langle}
\def\ra{\rangle}
\newcommand{\ip}[1]{{\langle #1 \rangle}}

\def\Span{\text{Span}}

\def\sse{\subseteq}

\newcommand{\initOneLiners}{%
    \setlength{\itemsep}{0pt}
    \setlength{\parsep }{0pt}
    \setlength{\topsep }{0pt}
}

%% file: intro.tex
\section{Introduction}

We consider the \emph{convex function chasing} (CFC) problem defined
by~\cite{FL93}, and independently studied under the name \emph{smooth online
  convex optimization} (SOCO) by~\cite{LLWA12,DBLP:journals/ton/LinWAT13}. In this problem, an
online player is faced with a sequence of convex functions over time,
and has to choose a good sequence of responses to incur small
function costs while also minimizing the movement cost for switching
between actions. Formally, the player starts at some initial
default action $x_0$, which is usually modeled as a point in $\R^d$. Convex functions $f_1, f_2, \dots$ arrive
online, one by one. Upon seeing the function $f_t: \R^d\to \R^+$, the player must
choose an action $x_t$. The cost incurred by this action is
\[ \|x_t - x_{t-1}\|_2 + f_t(x_t), \] the former
Euclidean distance term being the \emph{movement} or \emph{switching} cost between the
previous action $x_{t-1}$ and the current action $x_t$, and the latter
function value term being the \emph{hit cost} at this new action
$x_t$. 
(he problem can be defined for general metric spaces; in this paper we study the Euclidean case.) 
Given some sequence of functions
$\sigma = f_1, f_2, \ldots, f_T$, the online player's total cost for the associated
sequence $X = (x_1, x_2, \ldots, x_T)$ is
\begin{gather}
  \cost(X, \sigma) := 
  \sum_{t=1}^T \Big( \|x_t - x_{t-1}\|_2 + f_t(x_t) \Big). \label{eq:4}
\end{gather}
The \emph{competitive ratio} for this player is
$\max_{\sigma} \frac{\cost(ALG(\sigma), \sigma)}{\min_Y \cost(Y,
  \sigma)}$, the worst-case ratio of the cost of sequence of the
player when given request sequence $\sigma$, to the cost of the
optimal (dynamic) player for it (which is allowed to change its
actions but has to also pay for its movement cost). 
The goal is to give an online algorithm that has a small competitive
ratio.

The CFC/SOCO problem is usually studied in the setting where the action space is all of $\R^d$. We consider the generalized setting where the action space is any convex set $K\sse \R^d$. Formally, the set $K$ is fixed before the arrival of $f_1$, and each action $x_t$ must be chosen from $K$.

The CFC/SOCO problem captures many other problems arising
in sequential decision making. For instance, it can be used to model problems
in ``right-sizing'' data centers, charging electric cars, online logistic
regression, speech animation, and control; see, e.g., works by~\cite{LLWA12,
  WHLM14,KG14,GCW17,goel2019beyond} and the references therein. In all these
problems, the action $x_t$ of the player captures the state of the
system (e.g., of a fleet of cars, or of machines in a datacenter), and
there are costs associated both with taking actions at each timestep, and with
changing actions between timesteps. The CFC/SOCO problem models the challenge of trading
off these two costs against each other. 

One special case of CFC/SOCO is the \emph{convex body chasing}
problem, where the convex functions are indicators of convex sets in
$\R^d$. This special case itself captures the continuous versions of
problems in online optimization that face similar tensions between taking near-optimal actions and
minimizing movement: e.g.,~\emph{metrical task systems}
studied by \cite{BRS,BCLL19}, \emph{paging} and \emph{$k$-server}
(see~\cite{BCLLM18,BGMN19} for recent progress), and many others.

Given its broad expressive power, it is unsurprising that the
competitiveness of CFC/SOCO depends on the dimension $d$ of the space.
Indeed, \cite{FL93} showed a lower bound of $\sqrt{d}$ on the
competitive ratio for convex \emph{body} chasing, and hence for
CFC/SOCO as well. However, it was difficult to prove results about
the upper bounds: Friedman and Linial gave a constant-competitive
algorithm for \emph{body} chasing for the case $d=2$, and the
\emph{function} chasing problem was optimally solved for $d=1$ by
\cite{Bansal15}, but the general problem remained open for any higher
dimensions. The logjam was broken in results
by~\cite{Bansal,argue2019nearly} for some special cases, using ideas
from convex optimization. After intense activity since then,
algorithms with competitive ratio $O(\min(d, \sqrt{d \log T}))$ were
given for the general CFC/SOCO problem
by~\cite{AGGT19,Sellke19}. These results qualitatively settle the
question in the worst case---the competitive ratio is polynomial in
$d$---although quantitative questions about the exponent for $d$
remain.

However, this polynomial dependence on the dimension $d$ can be very pessimistic,
especially in cases when the convex functions have more structure. In these
well-behaved settings, we may hope to get better results and thereby escape this
curse of dimensionality. This motivates our work in this paper: we consider two such settings,
and give dimension-independent guarantees for them.

\paragraph{Well-Conditioned Functions.}
The first setting we consider is when the functions $f_t$ are all
\emph{well-conditioned convex} functions. Recall that a convex
function has \emph{condition number} $\kappa$ if it is
$\a$-\emph{strongly-convex} and $\b$-\emph{smooth} for some constants
$\a,\b > 0$ such that $\frac{\b}{\a}= \k$. Moreover, we are given a
convex set $K$, and each point $x_t$ we return must belong to $K$. (We
call this the \textrm{constrained} CFC/SOCO problem; while constraints
can normally be built into the convex functions, it may destroy the
well-conditionedness in our setting, and hence we consider it separately.)

Our first main result is the
following:
\begin{restatable}[Upper Bound: Well-Conditioned Functions]{theorem}{OBDmain}
  \label{thm:main}
  There is an $O(\sqrt{\kappa})$-competitive algorithm for constrained
  CFC/SOCO problem, where the functions have condition number at most $\k$.
\end{restatable}

Observe that the competitiveness does not depend on $d$, the dimension
of the space. Moreover, the functions can have very different
coefficients of smoothness and strong convexity, as long as their
ratio is bounded by $\k$. In fact, we give two algorithms. Our first
algorithm is a direct generalization of the greedy-like \emph{Move
  Towards Minimizer} algorithm of~\cite{Bansal15}. While it only
achieves a competitiveness of $O(\kappa)$, it is simpler and works for
a more general class of functions (which we called ``well-centered''),
as well as for all $\ell_p$ norms. Our second algorithm is a
constrained version of the \emph{Online Balanced Descent} algorithm of
\cite{CGW18}, and achieves the competitive ratio claimed
in~\Cref{thm:main}. We then show a lower bound in the same ballpark:

\begin{restatable}[Lower Bound: Well-Conditioned Functions]{theorem}{LBDmain}
  \label{thm:main-lbd}
  Any algorithm for 
  chasing convex functions with condition number at most $\k$
  must have competitive ratio at least $\Omega(\kappa^{1/3})$.
\end{restatable}

It remains an intriguing question to close the gap
between the upper bound of $O(\sqrt{\k})$ from \Cref{thm:main} and the
lower bound of $\Omega(\kappa^{1/3})$ from \Cref{thm:main-lbd}.
Since we show that $O(\kappa)$ and $O(\sqrt{\kappa})$ are
respectively tight bounds on the competitiveness of the two algorithms
mentioned above, closing the gap will require changing the algorithm.

\paragraph{Chasing Low-Dimensional Functions.} The second case is when
the functions are supported on low-dimensional subspaces of $\R^d$.
One such special case is when the functions are indicators of
$k$-dimensional affine
subspaces; this problem is referred to as chasing subspaces. 
If $k=0$ we are chasing
points, and the problem becomes trivial. \cite{FL93} gave a
constant-competitive algorithm for the first non-trivial case, that of
$k=1$ or \emph{line chasing}. \cite{Anto} simplified and improved this
result, and also gave an $2^{O(d)}$-competitive algorithm for chasing
general affine subspaces. Currently, the best bound even for
$2$-dimensional affine subspaces---i.e., planes---is $O(d)$, using the
results for general CFC/SOCO.

\begin{restatable}[Upper Bound: Low-Dimensional Chasing]{theorem}{SCCmain}
  \label{thm:main-sub}
  There is an $O(\min(k, \sqrt{k \log T}))$-compe\-titive algorithm for
  chasing convex functions supported on affine subspaces of dimension
  at most $k$.
\end{restatable}

The idea behind \Cref{thm:main-sub} is to perform a certain kind of
dimension reduction: we show that any instance of chasing
$k$-dimensional functions can be embedded into an $(2k+1)$-dimensional
instance, without changing the optimal solutions. Moreover, this
embedding can be done online, and hence can be used to extend any
$g(d)$-competitive algorithm for CFC/SOCO into a $g(2k+1)$-competitive
algorithm for $k$-dimensional functions.

\subsection{Related Work}
\label{sec:related-work}

There has been prior work on dimension-independent bounds for other
classes of convex functions. 
The Online Balanced Descent (\OBD) algorithm of \cite{CGW18} is
$\a$-competitive on Euclidean metrics if each function $f_t$ is $\a$-{locally-polyhedral} (i.e.,
it grows at least linearly as we go away from the minimizer). Subsequent works of \cite{GW19, goel2019beyond}
consider \emph{squared Euclidean distances} and give algorithms with
dimension-independent competitiveness of $\min(3+O(1/\alpha),
O(\sqrt{\alpha}))$  for $\alpha$-strongly
convex functions.
The requirement of squared Euclidean distances in these latter works
is crucial for their results: we show in \Cref{fct:lbd} that no online algorithm can
have dimension-independent competitiveness for \emph{non-squared}
Euclidean distances if the functions are only $\a$-strongly convex (or
only $\b$-smooth). Observe that 
our algorithms do not depend on the
actual value of the strong convexity coefficient $\a$, only on the ratio
between it and the smoothness coefficient $\beta$---so the functions
$f_t$ may have very different $\alpha_t, \beta_t$ values, and these
$\alpha_t$ may even be arbitrarily close to zero.

A related problem is the notion of regret minimization, which considers the additive gap of the
algorithm's cost~(\ref{eq:4}) with respect to the best \emph{static}
action $x^*$ instead of the multiplicative gap with respect to the
best \emph{dynamic} sequence of actions.
The notions of competitive ratio and regret are known to be 
inherently in conflict: \cite{ABSK13} showed that algorithms minimizing regret must
have poor competitive ratios in the worst-case. Despite this negative result, many 
ideas do flow from one setting to the other. 
These is a vast body of work where the algorithm is allowed to move
for free: see, e.g., books by~\cite{Seb,Hazan,SSS} for many
algorithmic ideas. This includes bounds comparing to the static
optimum, and also to a dynamic optimum with a bounded movement cost~\cite{Zinkevich03,BesbesGZ15,MSJR16,BubeckLLW19}.

Motivated by convergence and generalization bounds for learning
algorithms, the path length of gradient methods have been studied
by~\cite{OymakS19,GBR}. Results for CFC/SOCO also imply path-length
bounds by giving the same function repeatedly: the difference is that
these papers focus on a specific algorithm (e.g., gradient flow/descent),
whereas we design problem-specific algorithms (\MtoM or \COBD). 

The CFC/SOCO problem has been considered in the case with
\emph{predictions} or \emph{lookahead}: e.g., when the next $w$ functions are
available to the algorithm. For example,~\cite{LLWA12,LQL18} explore
the value of predictions in the context of data-server management, and
provide constant-competitive algorithms. For more recent work see,
e.g.,~\cite{LinGW19} and the references therein.

\subsection{Definitions and Notation}
\label{sec:definitions-notation}

We consider settings where the convex functions $f_t$ are
non-negative and differentiable. 
Given constants $\a,\b >0$, 
a differentiable function $f:\R^d\to \R$ is $\a$-\emph{strongly-convex} with respect to the norm $\|\cdot\|$ if 
for all $x,y\in \R^d$,
\[f(y) - f(x) - \la \grad f(x), y-x \ra \ge \frac{\a}{2}\|x-y\|^2,\]
and $\b$-\emph{smooth} if for all $x,y\in \R^d$,
\[f(y) - f(x) - \la \grad f(x), y-x \ra 
	\le \frac{\b}{2}\|x-y\|^2.  \]
A function $f$ is $\kappa$-\emph{well-conditioned} if there is a constant $\a>0$
for which $f$ is both $\a$-strongly-convex and $\alpha \kappa$-smooth.
Of course, we focus on the Euclidean $\ell_2$ norm (except briefly in \S\ref{sec:general-norm}), and hence $\|\cdot\|$ denotes $\|\cdot\|_2$ unless otherwise specified.

In the following, we assume that all our functions $f$ satisfy the
\emph{zero-minimum} property: i.e., that $\min_y f(y) = 0$. Else we can
consider the function $g(x) = f(x) - \min_y f(y)$ instead: this is
also non-negative valued, with the same smoothness and strong
convexity as $f$. Moreover, the competitive ratio can only increase
when we go from $f$ to $g$,
since the hit costs of both the algorithm and the optimum decrease by
the same additive amount.

\paragraph{Cost.}  Consider a sequence $\sigma = f_1, \ldots,
f_T$ of functions. If the algorithm moves from $x_{t-1}$ to $x_t$
upon seeing function $f_t$, the \emph{hit cost} is $f_t(x_t)$, and the
\emph{movement cost} is $\| x_t - x_{t-1}\|$.  Given a sequence
$X = (x_1,\dots, x_T)$ and a time $t$, define
$\cost_t(X,\sigma) := \|x_t - x_{t-1}\| + f_t(x_t)$ to be the total cost
(i.e., the sum of the hit and movement costs) incurred at time
$t$. When the algorithm and request sequence $\sigma$ are
clear from context, let $X_{ALG} = (x_1,x_2, \dots, x_T)$ denote the
sequence of points that the algorithm plays on $\sigma$.  Moreover, denote the
offline optimal sequence of points by
$Y_{OPT} = (y_1,y_2,\dots, y_T)$. For brevity, we omit mention of
$\sigma$ and let
$\cost_t(ALG) := \cost_t(X_{ALG}, \sigma)$ and
$\cost_t(OPT) := \cost_t(Y_{OPT},\sigma)$.

\paragraph{Potentials and Amortized Analysis.}
Given a potential $\Phi_t$ associated with time $t$, denote
$\Delta_t \Phi := \Phi_t - \Phi_{t-1}$. Hence, for all the amortized
analysis proofs in this paper, the goal is to show
\[ \cost_t(ALG) + a\cdot \Delta_t \Phi \leq b\cdot \cost_t(OPT) \]
for suitable parameters $a$ and $b$. Indeed, summing this over all
times gives
\[ \text{(total cost of }ALG) + a(\Phi_T - \Phi_0) \leq b \cdot
  \text{(total cost of }OPT). \]
Now if $\Phi_0 \leq \Phi_T$, which is the case for all our potentials,
we get that the cost of the algorithm is at most $b$ times the optimal
cost, and hence the algorithm is $b$-competitive.

\paragraph{Deterministic versus Randomized Algorithms.}
We only consider deterministic algorithms. This is without loss of
generality by the observation in \cite[Theorem~2.1]{Bansal15}:
given a randomized algorithm which plays the random point $X_t$ at
each time $t$, instead consider deterministically playing the ``average'' point $\mu_t := \E[X_t]$. 
This does not increase either the movement or the hit cost, 
due to Jensen's inequality and the convexity of the functions $f_t$ and the
norm $\|\cdot \|$.

%% file: algos.tex
\section{Algorithms}

We now give two algorithms for convex function chasing:
\S\ref{sec:move-to-min} contains the simpler \emph{Move Towards
  Minimizer} algorithm that achieves an $O(\k)$-competitiveness for 
  $\k$-well-conditioned functions, and a more general class of
 \emph{well-centered functions} (defined in~\Cref{sec:well-centered}). 
 Then \S\ref{sec:constrained-obd} contains the
\emph{Constrained Online Balanced Descent} algorithm that achieves the
$O(\sqrt{\k})$-competitiveness claimed in \Cref{thm:main}.

\subsection{Move Towards Minimizer: \texorpdfstring{$O(\kappa)$}{O(kappa)}-Competitiveness}
\label{sec:move-to-min}

The \emph{Move Towards Minimizer} (\MtoM) algorithm was defined in~\cite{Bansal15}.

\begin{quote}
  \textbf{The M2M Algorithm.} Suppose we are at position $x_{t-1}$ and receive the
  function $f_t$. Let $x^*_t := \arg\min_x f_t(x)$ denote the
  minimizer of $f_t$. Consider the line segment with endpoints
  $x_{t-1}$ and $x^*_t$, and let $x_t$ be the unique point on this
  segment with $\|x_t - x_{t-1}\| = f_t(x_t)-f_t(x_t^*)$.\footnote{Such a 
  point is always unique when $f_t$ is strictly convex.} 
  The point $x_t$ is the one
  played by the algorithm.
\end{quote}

The intuition behind this algorithm is that one of two things happens:
either the optimal algorithm $OPT$ is at a point $y_t$ near $x_t^*$,
in which case we make progress by getting closer to $OPT$. Otherwise,
the optimal algorithm is far away from $x_t^*$, in which case the hit
cost of $OPT$ is large relative to the hit cost of $ALG$.

\begin{figure}
\begin{center}
\includegraphics[scale=1]{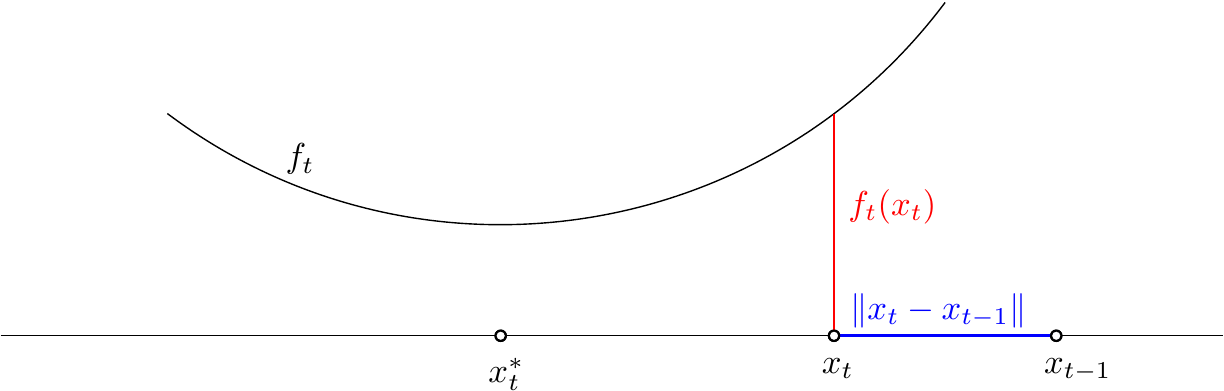}
\end{center}
\caption{The \MtoM Algorithm in dimension $d=1$.}
\label{fig:move-to-opt}
\end{figure}

As noted in \S\ref{sec:definitions-notation}, we assume that $f_t(x_t^*)=0$, 
hence \MtoM plays a point $x_t$ such that $\|x_t - x_{t-1}\| = f_t(x_t)$.
Observe that the total cost incurred by the algorithm at time $t$ is 
\[ \cost_t(ALG) = f_t(x_t) + \| x_t - x_{t-1}\| = 2f_t(x_t) =
  2\|x_t - x_{t-1}\|. \]

\subsubsection{The Analysis}

The proof of competitiveness for \MtoM
is via a potential function argument. The potential function captures the distance between the
algorithm's point $x_t$ and the optimal point $y_t$. Specifically, fix
an optimal solution playing the sequence of points
$Y_{OPT}=(y_1, \dots, y_T)$, and define
\[ \Phi_t:= \|x_t - y_t\|. \] Observe that $\Phi_0 = 0$ and
$\Phi_t \geq 0$. 

\begin{theorem}
  \label{thm:k-competitive}
  With $c := 4+4\sqrt 2$, for each $t$,
  \begin{gather}
    \cost_t(ALG) + 2\sqrt 2 \cdot \Delta_t \Phi \le c\cdot \k \cdot
    \cost_t(OPT). \label{eq:1}
  \end{gather}
  Hence, the \MtoM algorithm is $c \k$-competitive.
\end{theorem}

The main technical work is in the following lemma, which will be used
to establish the two cases in the analysis. Referring to
Figure~\ref{fig:structure-lemma}, imagine the minimizer for
$f_t$ as being at the origin, the point $y$ as being the location of
$OPT$, and the points $x$ and $\gamma x$ as being the old and new
position of $ALG$. Intuitively, this lemma says that either $ALG$'s
motion in the direction of the origin significantly reduces the
potential, or $OPT$ is far from the origin and hence has significant
hit cost.

\begin{lemma}[Structure Lemma]
  \label{lem:structure}
  Given any scalar $\g\in [0,1]$ and any two vectors $x,y\in \R^d$, at least one of the 
  following holds:
  \begin{enumerate}
  \item[(i)] $ \|y-\g x\| - \|y-x\| \le -\tfrac1{\sqrt 2} \|x-\g x\|$.
  \item[(ii)] $\|y\| \ge \tfrac{1}{\sqrt 2} \|\g x\|$.
  \end{enumerate}
\end{lemma}
\begin{proof}
  Let $\th$ be the angle between $x$ and $y-\g x$ 
  as in \Cref{fig:structure-lemma}. 
  If $\th < \frac{\pi}{2}$, then $\|y\| \ge \|\gamma x\|$, 
  and hence condition (ii) is satisfied.
  So let $\theta \in [\frac{\pi}{2}, \pi]$; using Figure~\ref{fig:structure-lemma} observe that
  \begin{equation}\label{eq:structure-lemma-1} 
  \|y\| \ge \sin(\theta)\cdot\|\gamma x\|.
  \end{equation}

\begin{figure}
\begin{center}
\includegraphics[scale=1]{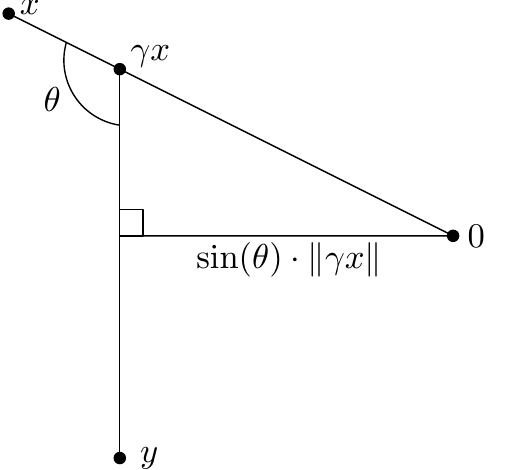}
\end{center}
\caption{The Proof of Lemma~\ref{lem:structure}.} 
\label{fig:structure-lemma}
\end{figure}

  Suppose condition (i) does not hold. Then
  \[ \|y-x\|  <  \tfrac1{\sqrt 2} \|(1-\g) x\| + \|y-\g x\|. \]
  Since both sides are non-negative, we can square to get
  \begin{align*}
    \|y-x\|^2  
    &< \frac12(1-\g)^2 \|x\|^2 + \sqrt{2}\cdot\|(1-\g)x\| \cdot \|y-\g x\|
    	+ \|y-\g x\|^2\\
    \implies \|y-x\|^2 - \|y-\g x\|^2  
    &< \frac12(1-\g)^2 \|x\|^2 + \sqrt2 (1-\g)\cdot \|x\|\cdot \|y-\g x\|.
  \end{align*}
  The law of cosines gives 
  \[\|y-x\|^2 - \|y-\g x\|^2  
  	= (1-\g)^2\|x\|^2 - 2(1-\g) \cos(\th)\cdot \|x\|\cdot  \|y-\g x\|.\]
  Substituting and simplifying,
  \[ \frac12 (1-\g)\| x\| < (\sqrt 2 +2\cos(\th)) \|y-\g x\|.\]
  As the LHS is non-negative, $\cos(\th) > -\frac1{\sqrt 2}$. 
  Since $\theta \ge \frac{\pi}{2}$, it follows that 
  $\sin(\th) > \frac{1}{\sqrt{2}}$.
  Now, (\ref{eq:structure-lemma-1}) implies 
  that $\|y\| \ge \sin(\theta)\cdot \| \gamma x\| \ge \frac{1}{\sqrt2}\|\g x\|$. 
\end{proof}

\begin{proof}[Proof of \cref{thm:k-competitive}]
  First, the change in potential can be bounded as
  \[ \Delta_t \Phi = \| x_t - y_t \| - \| x_{t-1} - y_{t-1} \| \leq \| x_t - y_t \| - \Big( \| x_{t-1} - y_t \| - \| y_t - y_{t-1} \| 
    \Big). \]
  The resulting term $\| y_t - y_{t-1} \|$ can be charged to the
  movement cost of $OPT$, and hence it suffices to show that
  \begin{gather}
    \cost_t(ALG) + 2\sqrt 2\cdot \DTPhi \le (4+4\sqrt 2)\k\cdot
    f_t(y_t), \label{eq:2}
  \end{gather}
  where $\DTPhi:= \| x_t - y_t \| - \| x_{t-1} - y_t \|$ 
  denotes the change in potential due to the movement of $ALG$.
  Recall that $x^*_t$ was the minimizer of the function $f_t$. The
  claim is translation invariant, so assume $x^*_t = 0$. This implies that $x_t = \g x_{t-1}$ for some
  $\g\in (0,1)$.  Lemma~\ref{lem:structure} applied to $y = y_t$,
  $x = x_{t-1}$ and $\g$, guarantees that one of the following holds:
  \begin{enumerate}
  \item[(i)] $\DTPhi = \| x_t - y_t \| - \| x_{t-1} - y_t \| \le -\tfrac1{\sqrt 2} \|x_t - x_{t-1}\|$.
  \item[(ii)] $\|y_t\| \ge \tfrac{1}{\sqrt 2} \|x_t\|$.
  \end{enumerate}

  \medskip\textbf{Case I:} Suppose $\DTPhi \le -\tfrac1{\sqrt 2} \|x_t - x_{t-1}\|$.
  Since $\cost_t(ALG) \le 2\|x_t-x_{t-1}\|$,
  \begin{align*}
 	\cost_t(ALG) + 2\sqrt 2\cdot \DTPhi
 	 &\le 2\|x_t-x_{t-1}\| - 2\|x_t-x_{t-1}\| 
 	 = 0 \\
 	 &\le (4+4\sqrt 2)\k\cdot f_t(y_t).
  \end{align*}
  This proves~(\ref{eq:2}).

\medskip\textbf{Case II:} Suppose that $\|y_t\| \ge \tfrac{1}{\sqrt 2} \|x_t\|$. 
By the well-conditioned assumption on $f_t$ (say, $f_t$ is $\a_t$-strongly-convex and $\a_t \k$ smooth) and the assumption that  
$0$ is the minimizer of $f_t$, we have
\begin{equation}\label{eq:k-competitive-2}
	f_t(x_t)
		\le \frac{\alpha_t \kappa}{2} \|x_t\|^2
		\le \alpha_t\kappa \|y_t\|^2
		\le 2 \kappa \cdot f_t(y_t).
\end{equation}
By the triangle inequality and choice of $x_t$ such that $f_t(x_t) = \|x_t - x_{t-1}\|$ we have
\[
\DTPhi =\| x_t - y_t \| -  \| x_{t-1} - y_t \|
	 \leq \|x_t - x_{t-1}\| = f_t(x_t).
\]
Using $\cost_t(ALG) = 2f_t(x_t)$,
\begin{align*}
 \cost_t(ALG) + 2\sqrt 2\cdot \DTPhi
	&\le 2f_t(x_t) + 2\sqrt 2 f_t(x_t)\\
	&\stackrel{(\ref{eq:k-competitive-2})}{\le} (4+4\sqrt 2) \k\cdot  f_t(y_t).
\end{align*}
This proves~(\ref{eq:2}) and hence the bound~(\ref{eq:1}) on the
amortized cost. Now summing~(\ref{eq:1}) over all times $t$, and using
that $\Phi_t \geq 0 = \Phi_0$, proves the competitiveness.
\end{proof}

We extend Theorem~\ref{thm:k-competitive} to the constrained 
setting (by a modified algorithm); see \S\ref{sec:constrained-m2m}. We also
extend the result to general norms by replacing 
Lemma~\ref{lem:structure} by Lemma~\ref{lem:general-norm}; details 
appear in \S\ref{sec:general-norm}. 
Moreover, the analysis of \MtoM is tight: 
in Proposition~\ref{prop:m2m-lbd} we show an instance for 
which the \MtoM algorithm has $\Omega(\k)$-competitiveness. 

\subsubsection{Well-Centered Functions}
\label{sec:well-centered}

The proof of \Cref{thm:k-competitive} did not require the full
strength of the well-conditioned assumption. In fact, it only required
that each function $f_t$ is $\k$-well-conditioned ``from the perspective of
its minimizer $x_t^*$'', namely that there is a constant $\a$ such that for all 
$x\in \R^d$,
\[\frac{\a}{2}\|x-x_t^*\|^2 \le f_t(x) \le \frac{\k\a}{2} \|x-x_t^*\|^2. \]
Motivated by this observation, we define a somewhat more general class of functions for which the \MtoM algorithm is competitive.

\begin{definition}
Fix scalars $\k,\g \ge 1$. A convex function $f: \R^d\to \R^+$ with minimizer $x^*$ is \emph{$(\k,\g)$-well-centered} if there is a constant $\a > 0$ such that for all $x\in \R^d$, 
\begin{equation*}
  \frac{\a}{2}\|x-x^*\|^\g \le f(x) \le \frac{\a\k}{2}\|x-x^*\|^\g. \label{eq:def-1}
\end{equation*}
\end{definition}

We can now give a more general result.

\begin{proposition}
  \label{cl:blah}
If each function $f_t$ is $(\k,\g)$-well centered, then with $c = 2 + 2\sqrt 2$, 
\[\cost_t(ALG) + 2\sqrt 2\cdot \Delta_t\Phi 
	\le c\cdot 2^{\g/2}\k \cdot \cost_t(OPT).\]
Hence, the \MtoM algorithm is $c2^{\g/2}\k$-competitive.
\end{proposition}
\begin{proof}
 Consider the proof of \Cref{thm:k-competitive} and replace
(\ref{eq:k-competitive-2}) by
\[	f_t(x_t)
		\le \frac{\alpha_t \kappa}{2} \|x_t\|^\gamma
		\le \frac{\alpha_t\kappa}{2} \|y_t\|^\gamma \cdot 2^{\gamma / 2}
		\le 2^{\gamma / 2} \kappa \cdot f_t(y_t).
              \]
              The rest of the proof remains unchanged.
\end{proof}

\subsection{Constrained Online Balanced Descent: \texorpdfstring{$O(\sqrt{\kappa})$}{O(sqrt(kappa))}-Competitiveness}
\label{sec:constrained-obd}

The move-to-minimizer algorithm can pay a lot in one 
timestep if the function decreases slowly in the direction of the 
minimizer but decreases quickly in a different direction. 
In the unconstrained setting, the \emph{Online Balanced Descent}
algorithm addresses this by moving to a point $x_t$ such that $\|x_t - x_{t-1}\| = f_t(x_t)$, except it chooses the 
point $x_t$ to minimize $f_t(x_t)$. It therefore minimizes the instantaneous cost $\cost_t(ALG)$ among all algorithms that balance the movement and hit costs. This algorithm can be viewed
geometrically as projecting the point $x_{t-1}$ onto a level set of
the function $f_t$; see~\Cref{fig:obd-vs-mtm}. 

\begin{figure}
\begin{center}
\includegraphics[scale=1]{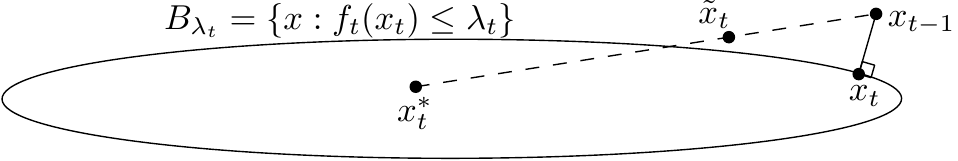}
\end{center}
\caption{The \OBD Algorithm and the comparison to \MtoM.
The point $x_{t-1}$ and the function $f_t$ with minimizer $x_t^*$ are given.
\OBD plays the point $x_t$ and \MtoM plays the point $\tilde{x}_t$.}
\label{fig:obd-vs-mtm}
\end{figure}

In the constrained setting, it may be the case that $\|x_t - x_{t-1}\| < f_t(x_t)$ for all feasible points. Accordingly, the \emph{Constrained Online Balanced Descent} (\COBD) algorithm moves to a point $x_t$ that minimizes $f_t(x_t)$ subject to $\|x_t - x_{t-1}\| \le f_t(x_t)$.

Formally, suppose that each $f_t$ is $\a_t$-strongly convex and $\b_t := \k
\a_t$-smooth, and let $x_t^*$ be the (global) minimizer of $f_t$,
which may lie outside $K$. As before, we
assume that $f_t(x_t^*) = 0$. 

\begin{quote}
  \textbf{The Constrained OBD Algorithm.} 
  Let $x_t$ be the solution to the (nonconvex) program $\min\{f_t(x) \mid
  \|x-x_{t-1}\| \le f_t(x), x\in K\}$. Move to the point
  $x_t$. (Regarding efficient implementation of \COBD, see 
  Remark~\ref{rem:efficient-cobd}.)
\end{quote}

As with \MtoM, the choice of $x_t$ such that $\|x_t-x_{t-1}\| \le f_t(x_t)$ implies that 
\[ \cost_t(ALG) = f_t(x_t) + \| x_t - x_{t-1}\| \le 2f_t(x_t).\]

\subsubsection{The Analysis.} 
Again, consider the potential function:
\begin{gather}
  \Phi_t := \| x_t - y_t \| \label{eq:pot}
\end{gather}
where $x_t$ is the point controlled
by the \COBD algorithm, and $y_t$ is the point controlled by the optimum
algorithm. We first prove two useful lemmas. The first lemma is a general 
statement about $\b$-smooth functions that is independent of the
algorithm.
\begin{lemma}\label{lem:grad-bound}
  Let convex function $f$ be $\b$-smooth. Let $x^*$ be the global 
  minimizer of $f$, and suppose $f(x^*) = 0$ (as discussed in 
  \S\ref{sec:definitions-notation}). Then for all $x\in \R^d$,
  \[\|\grad f(x) \| \le \sqrt{2\b f(x)}.\]
\end{lemma}
\begin{proof}
  The proof follows~\cite[Lemma 3.5]{Seb}. Define $z := x -\frac1{\b} \gr f(x)$. 
  Then
  \begin{align*}
    f(x) &\ge f(x) - f(z) \tag{since $f(z) \geq 0$}\\
    &\geq \ip{\gr f(x), x-z} - \frac{\b}{2}\|x-z\|^2 
      \tag{by $\b$-smoothness} \\
    &= \ip{\gr f(x), \frac1{\b} \gr f(x)} - \frac{1}{2\b} \| \gr f(x)
      \|^2 = \frac{1}{2\b} \| \gr f(x)
      \|^2.
  \end{align*}
  The conclusion follows.
\end{proof}

The second lemma is specifically about \COBD.

\begin{lemma}\label{lem:conic-comb}
For each $t \ge 1$, there is a constant $\lambda \ge 0$ and a vector $n$ in the normal cone to $K$ at $x_t$ such that $x_{t-1} - x_t = \lambda \grad f_t(x_t) + n$.
\end{lemma}
\begin{proof}
Let $r = \|x_t - x_{t-1}\|$.
We claim that $x_t$ is the solution to the following \emph{convex} program:
\begin{align*}
\min \quad & f_t(x)\\
\text{s.t.} \quad  & \|x - x_{t-1}\|^2 \le r^2\\
& x\in K
\end{align*}
Given this claim, the KKT conditions imply that there is a constant $\g\ge 0$ such that $\grad f_t(x) + \g (x_t - x_{t-1})$ is in the normal cone to $K$ at $x_t$ and the result follows.

We now prove the claim.
Assume for a contradiction that the solution to this program is a
point $z\neq x_t$. We have $f_t(z)< f_t(x_t)$. Since $z\in K$ and
$x_t$ is the optimal solution to the nonconvex program $\min\{f_t(x) \mid \|x-x_{t-1}\| \le f_t(x), x\in K\}$,
 we have $f(z) < \|z-x_{t-1}\|$. But considering the line segment with endpoints $z$ and $x_{t-1}$, the intermediate value theorem implies that there is a point $z'$ on this segment such that $f(z') = \|z'-x_{t-1}\|$. This point $z'$ is feasible for the nonconvex program and \[f(z') = \|z'-x_{t-1}\| < \|z-x_{t-1}\| = f(z) < f(x).\]
This contradicts the choice of $x_t$. The claim is proven, hence
the proof of the lemma is complete.
\end{proof}

\begin{remark}\label{rem:efficient-cobd}
The convex program given in the proof can be used to to find $x_t$ efficiently. In particular, let $r^*$ denote the optimal value to the nonconvex program. For a given $r$, if the solution to the convex program satisfies $f_t(x) < r$, then $r^* < r$. Otherwise, $r^* \ge r$. Noting that $0\le f_t(x_t)\le f_t(x_{t-1})$, run a binary search to find $r^*$ beginning with $r = \frac12 f_t(x_{t-1})$.
\end{remark}

\begin{theorem}
  \label{thm:sqrt-k-competitive}
  With $c = 2\sqrt{2\k}$, for each time $t$ it holds that
  \[\cost_t (ALG) + c\cdot \Delta_t \Phi 
    \leq 2(2+c)\cdot \cost_t(OPT).\]
  Hence, the \COBD algorithm is $2(2+c) = O(\sqrt{\k})$-competitive.
\end{theorem}
\begin{proof}
  As in the proof of~\Cref{thm:k-competitive}, 
  it suffices to show that
  \begin{gather}
    \cost_t(ALG) + c\cdot \DTPhi \le 2(2+c)\cdot f_t(y_t), 
  \end{gather}
  where $\DTPhi:= \|x_t - y_t\| - \|x_{t-1} - y_t\|$ is the change in potential
  due to the movement of $ALG$.  
  
  There are two cases, depending on the value of $f_t(y_t)$ versus the value of
  $f_t(x_t)$. In the first case, $f_t(y_t) \ge \frac{1}{2}
  f_t(x_t)$. The triangle inequality bounds 
  $\DTPhi = \| x_t - y_t \| - \| x_{t-1} - y_t \|\le \| x_t - x_{t-1}\|
   \le f_t(x_t)$. 
  Also using $\cost_t(ALG) \le 2f_t(x_t)$, we have
  \begin{gather*}
   \cost_t(ALG) + c\cdot \DTPhi
    	\le 2f_t(x_t) + c f_t(x_t)
    	\le 2(2+c)\cdot f_t(y_t). 
  \end{gather*}
	
  In the other case, $f_t(y_t) \le \frac{1}{2} f_t(x_t)$. 
  Note that this implies that $x_t$ is not the minimizer of $f_t$ on the 
  set $K$. Any move in the direction of the minimizer gives a point in $K$ with lower hit cost, but this point cannot be feasible for the nonconvex program. Therefore, at the point $x_t$, the constraint relating the hit cost to the movement cost is satisfied with equality: 
  $\|x_t-x_{t-1}\| = f_t(x_t)$.
  
  Let $\theta_t$ be the angle formed by the vectors $\grad f_t(x_t)$ 
  and $y_t-x_t$; see~\Cref{fig:obd-analysis}.
We now have
  \begin{align*}
    - \la \grad f_t(x_t), y_t-x_t\ra 
	&\ge f_t(x_t) - f_t(y_t) + \frac{\a_t}{2}\|x_t-y_t\|^2 \tag{by strong convexity}\\
	& \ge \frac{1}{2} f_t(x_t)  + \frac{\a_t}{2}\|x_t-y_t\|^2 \tag{since $f(y_t) \le \frac{1}{2} f(x_t)$}\\
\implies - \cos \theta_t
	&\ge \frac{\frac{1}{2}(f_t(x_t) + \a_t\|x_t-y_t\|^2)}{\|\grad f_t(x_t)\|\cdot \|x_t-y_t\|}\\
	&\ge \frac{\frac{1}{2}(f_t(x_t) +
   \a_t\|x_t-y_t\|^2)}{\sqrt{2\a_t \k \; f_t(x_t)}\cdot \|x_t-y_t\|} 
	\tag{by Lemma~\ref{lem:grad-bound}}\\
	&\ge \frac{1}{\sqrt{2 \k}}  \tag{by the AM-GM inequality}
\end{align*}

  By Lemma~\ref{lem:conic-comb}, we have $x_{t-1} - x_t = \lambda \grad f_t(x_t) + n$ for some $n$ in the normal cone to $K$ at point $x_t$. Since $y_t\in K$ we have $\la n, y_t - x_t\ra \le 0$. This gives
  \begin{equation} \label{eq:thm10-1}
  - \la x_{t-1} - x_t, y_t - x_t \ra
  	= - \la \lambda \grad f_t(x_t) + n, y_t - x_t \ra
	\ge - \lambda \grad \la f_t(x_t), y_t - x_t\ra
  \end{equation}
  Furthermore, we have $\lambda \grad f_t(x_t) = (x_{t-1} - x_t) - n$, and since $\la x_{t-1} - x_t, n \ra < 0$ we have
  \begin{equation}\label{eq:thm10-2}
   \|x_{t-1} - x_t\| \le \lambda \|\grad f_t(x_t)\|
  \end{equation}
  Let $\varphi_t$ be the angle formed by the vectors $x_{t-1}-x_t$ 
  and $y_t-x_t$; see~\Cref{fig:obd-analysis}.
  \begin{figure}[b]
  \begin{center}
	\includegraphics[scale=.8]{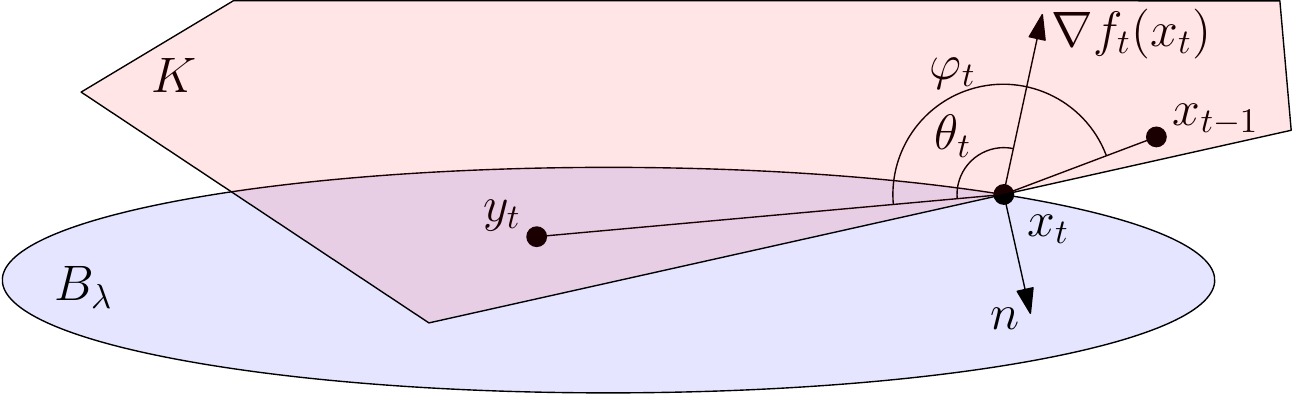}
  \end{center}
  \caption{Proof of Theorem~\ref{thm:sqrt-k-competitive}, case when 
  $f_t(y_t) \le \frac{1}{2}f_t(x_t)$. $B_\lambda$ is the sublevel set of $f_t$ with $x_t$ is on its boundary. }
  \label{fig:obd-analysis}
\end{figure}

Combining the previous three inequalities,
  \begin{align*}
  -\sec\varphi_t
  	&= \frac{\|x_{t-1} - x_t\|\cdot \|y_t - x_t\|}{-\la x_{t-1}-x_t, y_t - x_t \ra }\\
	&\le \frac{\lambda \|\grad f_t(x_t)\|\cdot \|y_t - x_t\|}{-\lambda \la \grad f_t, y_t - x_t\ra} \tag{by (\ref{eq:thm10-1}), (\ref{eq:thm10-2})}\\
	&= -\sec \th_t \\
	&\le \sqrt{2\kappa} = \frac{c}{2}
  \end{align*}
Now the law of cosines gives:
  \[\|x_t - x_{t-1}\|^2 - 2\|x_t - x_{t-1}\|\cdot \|x_t-y_t\|\cos \varphi_t
  =  \|x_{t-1}-y_t\|^2 - \|x_t - y_t\|^2.\]
Rearranging:
\begin{align*}
\|x_t - x_{t-1}\|
	&=  \left(\frac{\|x_{t-1} - y_t\| + \|x_t - y_t\|}{\|x_t - x_{t-1}\| - 2\|x_t-y_t\| \cos \varphi_t} \right)
	\Big(\|x_{t-1} - y_t\| - \|x_t - y_t\|\Big)\\
	&\le  \left(\frac{\|x_t - x_{t-1}\| + 2\|x_t-y_t\|}{\|x_t - x_{t-1}\| - 2\|x_t-y_t\| \cos \varphi_t} 
	\right) \Big(\|x_{t-1} - y_t\| - \|x_t - y_t\|\Big) \tag{triangle inequality}\\
	& \le -(\sec \varphi_t) \cdot \Big(\|x_{t-1} - y_t\| - \|x_t - y_t\|\Big).
\end{align*}
To see the last inequality, recall that $-\cos \varphi_t > 0$;
hence $\frac{a+b}{a+b(-\cos \varphi_t)} 
\le \frac{a+b}{(a+b)(-\cos \varphi_t)} = -\sec \varphi_t$.
Using that $\cost_t(ALG) = 2\|x_t - x_{t-1}\|$, we can rewrite the 
inequality above as 
\[\cost_t(ALG) 
	- (2\sec \varphi_t) \cdot \DTPhi \le 0.\]
Finally, observe that since $y_t\in B_{\lambda_t}$, 
we have $\DTPhi \le 0$.
Using the fact that $-\sec(\varphi_t) \le \frac{c}{2}$, 
\[\cost_t(ALG) + c\DTPhi \le \cost_t(ALG) - (2\sec \varphi_t) \cdot \DTPhi
	\le 0 \le 2(2+c)\cdot  f_t(y_t).\]
This completes the proof.
\end{proof}

Again, our analysis of \COBD is tight: In Proposition~\ref{prop:obd-lbd} we show an instance for which the \COBD algorithm has $\Omega(\sqrt\k)$-competitiveness, even in the unconstrained setting.

%% file: low-dim.tex
\def\Span{\text{span}}

\section{Chasing Low-Dimensional Functions}
\label{sec:subspaces}

In this section we prove \Cref{thm:main-sub}, our main result for
chasing low-dimensional convex functions. We focus our attention to
the case where the functions $f_t$ are indicators of some affine
subspaces $K_t$ of dimension $k$, i.e., $f_t(x) = 0$ for $x \in K_t$
and $f_t(x) = \infty$ otherwise. (The extension to the case where we
have general convex functions supported on $k$-dimensional affine
subspaces follows the same arguments.)
The main ingredient in the proof of chasing low-dimensional affine
subspaces is the following dimension-reduction theorem:

\begin{theorem}
  \label{thm:reduction}
  Suppose there is an $g(d)$-competitive algorithm for chasing convex
  bodies in $\R^d$, for each $d \geq 1$. 
  Then for any $k \leq d$, there is a
  $g(2k+1)$-competitive algorithm to solve instances of chasing convex
  bodies in $\R^d$ where each request lies in an affine subspace of
  dimension at most $k$.
\end{theorem}

In particular, \Cref{thm:reduction} implies that there is a
$(2k+1)$-competitive algorithm for chasing subspaces of dimension at
most $k$, and hence proves \Cref{thm:main-sub}.\\

\begin{proof}
  Suppose we have a chasing convex bodies instance
  $K_1, K_2, \dots, K_T$ such that each $K_t$ lies in some
  $k$-dimensional affine subspace. We construct another sequence
  $K_1', \dots, K_T'$ such that (a)~there is a single $2k+1$
  dimensional linear subspace $L$ that contains each $K'_t$, and
  (b) there is a feasible point sequence $x_1,\dots, x_T$ of cost $C$ for the
  initial instance if and only if there is a feasible point sequence
  $x'_1,\dots, x'_T$ for the transformed instance with the same cost.
  We also show that
  the transformation from $K_t$ to $K_t'$, and from $x_t'$ back to
  $x_t$ can be done online, resulting in the claimed algorithm.

  Let $\Span(S)$ denote the affine span of the set $S\subseteq
  \R^d$. Let $\dim(A)$ denote the dimension of an affine subspace
  $A\subseteq \R^d$. The construction is as follows: let $L$ be an
  arbitrary $(2k+1)$-dimensional linear subspace of $\R^d$ that
  contains $K_1$. We construct online a sequence of affine isometries
  $R_1,\dots, R_T$ such that for each $t > 1$:
  \begin{enumerate}
  \item[(i)] $R_t(K_t)\subseteq L$.
  \item[(ii)] $\|R_t(x_t) - R_{t-1}(x_{t-1}) \| = \|x_t - x_{t-1}\|$ for any 
    $x_{t-1}\in K_{t-1}$ and $x_t\in K_t$.
  \end{enumerate}
  Setting $x_t' = R_t(x_t)$ then achieves the goals listed above. To
  get the affine isometry $R_t$ we proceed inductively: let $R_1$ be the identity map,
  and suppose we have constructed $R_{t-1}$.  Let
  $A_t:= \Span(R_{t-1}(K_t) \cup R_{t-1}(K_{t-1}))$. Note that
  $\dim(A_t) \le 2k+1$. Let $\rho_t$ be an affine isometry that fixes
  $\Span(R_{t-1}(K_{t-1}))$ and maps $\Span(R_{t-1}(K_t))$ into
  $L$. Now define $R_t = \rho_t\circ R_{t-1}$.  Property~(i) holds by
  construction. Moreover, since $x_{t-1}\in K_{t-1}$, we have
  $R_t(x_{t-1}) = R_{t-1}(x_{t-1})$. Furthermore, $R_t$ is an isometry
  and hence preserves distances. Thus,
  \[\|R_t(x_t) - R_{t-1}(x_{t-1}) \| 
    = \|R_t(x_t) - R_t(x_{t-1}) \| = \|x_t - x_{t-1}\|. \]
  This proves (ii).

  Note that $R(x_1,\dots, x_T) := (R_1(x_1), \dots, R_T(x_T))$ is a
  cost-preserving bijection between point sequences that are feasible
  for $\{K_t\}_t$ and $\{K_t'\}_t$ respectively. It now follows that
  the instances $\{K_t\}_t$ and $\{K_t'\}_t$ are equivalent in the
  sense that $OPT(K_1',\dots, K_T') = OPT(K_1,\dots, K_T)$, and an
  algorithm that plays points $x'_t\in K'_t$ can be converted into an
  algorithm of equal cost that plays points $x_t\in K_t$ by letting
  $x_t = R_t^{-1}(x'_t)$.  However, each of $K_1', \dots, K_T'$ is
  contained in the $(2k+1)$ dimensional subspace $L$, and thus we get
  the $g(2k+1)$-competitive algorithm.
\end{proof}

Using the results for CFC/SOCO, this immediately gives an
$O(\min(k, \sqrt{k \log T}))$-competitive algorithm to chase convex
bodies lying in $k$-dimensional affine subspaces. Moreover, the lower
bound of \cite{FL93} immediately extends to show an $\Omega(\sqrt{k})$
lower bound for $k$-dimensional subspaces. Finally, the proof for
$k$-dimensional functions follows the same argument, and is deferred
for now.

%% file: lower-bounds.tex
\section{Lower Bounds}
\label{sec:lower-bounds}

In this section, we show a lower bound of $\Omega(\k^{1/3})$ on the
competitive ratio of convex function chasing for $\k$-well-conditioned
functions. We also show that our analyses of the \MtoM and \COBD
algorithms are tight: that they have competitiveness $\Omega(\k)$ 
and $\Omega(\sqrt{\k})$ respectively. In both examples, we take
$K=\R^d$ to be the action space.

\subsection{A Lower Bound of \texorpdfstring{$\Omega(\k^{1/3})$}{Omega(cube-root-kappa)}}

The idea of the lower bound is similar to the $\Omega(\sqrt{d})$ lower
bound~\cite{FL93}, which we now sketch. In this lower bound, the
adversary eventually makes us move from the origin to some vertex
$\pmb{\e} = (\e_1, \e_2, \ldots, \e_d)$ of the hypercube
$\{-1,1\}^d$. At time $t$, the request $f_t$ forces us to move to the
subspace $\{ x \mid x_i = \e_i\; \forall i \leq t \}$. Not knowing the
remaining coordinate values, it is best for us to move along the
coordinate directions and hence incur the $\ell_1$ distance of
$d$. However the optimal solution can move from the origin to
$\pmb{\e}$ along the diagonal and incur the $\ell_2$ distance of
$\sqrt{d}$. Since the functions $f_t$ in this example are not well-conditioned, we
approximate them by well-conditioned functions; however, this causes
the candidate $OPT$ to also incur nonzero hit costs, leading to the
 lower bound of $\Omega(\k^{1/3})$ when we balance the hit and
 movement costs.

We begin with a lemma analyzing a general instance defined by several parameters, 
and then achieve multiple lower bounds by appropriate choice of the parameters.

\begin{lemma}
  \label{lem:lbd}
  Fix a dimension $d$, constants $\g >0$ and $\lambda \ge \mu \ge 0$.
  Given any algorithm $ALG$ for chasing convex functions, there is a request 
  sequence $f_1, f_2, \dots, f_d$ that satisfies: \vspace{-4pt}
  \begin{enumerate}
	\item[(i)] Each $f_t$ is $2\mu$-strongly-convex and $2\lambda$-smooth (hence
	$(\lambda/\mu)$-well-conditioned.)
	\item[(ii)] $OPT \le \g(1+\mu d^{3/2} \g) \sqrt{d}$.
	\item[(iii)] $ALG \ge (\gamma - \frac{1}{4\lambda})d$.
  \end{enumerate}
\end{lemma}
\begin{proof}
  Consider the instance where at each time $t \in \{1, \ldots, d\}$, we
  pick a uniformly random value $\e_t \in \{-1,1\}$, and set
  \[f_t(x) = \lambda \sum_{i=1}^t (x_i-\g \e_i)^2 + \mu\sum_{i=t+1}^d x_i^2.\]
  One candidate for $OPT$ is to move to the point
  $\g\pmb{\e} := (\g \e_1, \g\e_2, \ldots, \g\e_d)$, and take all the functions
  at that point. The initial movement costs $\g\sqrt{d}$, and the $t^{\text{th}}$
  timestep costs $f_t(\g\pmb{\e}) = \mu(d-t)\g^2$. Hence, the total cost over the
  sequence is at most 
  \[\g\sqrt{d} + \mu\binom{d}{2}\g^2
  	\leq \g\Big(1+\mu d^{3/2} \g\Big) \sqrt{d}.\]  
  Suppose the algorithm is at the point $\mathbf{z} = (z_1, \ldots, z_d)$ 
  after timestep $t-1$, and it moves to point 
  $\mathbf{z}' = (z'_1, \ldots, z'_d)$ at the next timestep. 
  Moreover, suppose the algorithm sets $z'_t=a$ when it sees $\e_t = 1$, and
  sets $z'_t = b$ if $\e_t=-1$. 
  Then for timestep $t$, the algorithm pays in expectation at least
  \begin{align*}
  &\frac12 [\lambda(a - \g)^2 + |a - z_t|] 
  		+  \frac12 [\lambda(b+\g)^2 + | b - z_t| ] \\
  &= \frac{\lambda}{2} \left[(a^2 - 2\g a + \g^2) 
  		+ (b^2 + 2\g b + \g^2) \right] 
  		+  \frac12[|a - z_t| + | b - z_t| ] \\
  &\geq \frac{\lambda}{2} \left[(a^2 - 2\g a + \g^2) 
  		+ (b^2 + 2\g b + \g^2) \right] + \frac12 (a -b)  \\
  &= \frac{\lambda}{2} 
  		\left[\left(a^2 - \left(2\g-\frac{1}{\lambda}\right) a + \g^2\right) 
  		+ \left(b^2 + \left(2\g-\frac{1}{\lambda}\right) b + \g^2\right) \right]  \\
  &\geq \g - \frac{1}{4\lambda}.
  \end{align*}
  The last inequality follows from choosing $a = \g - 1/(2\l)$ and 
  $b = \g + 1/(2\l)$ to minimize the respective
  quadratics. Hence, in expectation, the algorithm pays at least $\gamma - \frac{1}
  {4\lambda}$ at each time $t$. Summing over all times, we get a lower bound of
  $(\gamma - \frac{1}{4\lambda})d$ on the algorithm's cost.  
\end{proof}

In particular, Lemma~\ref{lem:lbd} implies a competitive ratio  of at least
\[\left(\frac{\gamma - 1/(4\lambda)}{\gamma(1+\mu d^{3/2} \gamma)}\right)\sqrt d\]
for chasing a class of functions that includes $f_1,\dots , f_d$.
It is now a simple exercise in choosing constants to get a lower bound on the 
competitiveness of any algorithm for chasing $\k$-well-conditioned functions, 
$\a$-strongly-convex functions, and $\b$-smooth functions.

\begin{proposition}
  \label{fct:lbd}  
          The competitive ratio of any algorithm for chasing
          convex functions with condition number $\k$ is $\Omega(\k^{1/3})$.
Moreover,          the competitive ratio of any algorithm for chasing 
  	$\alpha$-strongly-convex (resp., $\b$-smooth) functions is $\Omega(\sqrt{d})$.
\end{proposition}
\begin{proof}
	For $\k$-strongly convex functions,	
	apply Lemma~\ref{lem:lbd} with dimension $d = \k^{2/3}$, constants 
	$\gamma = \lambda = 1$ and $\mu = \kappa^{-1} = d^{-3/2}$.
	This shows a gap of $\Omega(\sqrt{d}) = \Omega(\k^{1/3})$.
	For $\a$-strongly convex functions, choose
	$\mu = \alpha/2$, $\g = 1/(d^{3/2} \alpha)$, and 
	$\lambda = 1/\g = d^{3/2}\alpha$. Finally, 	
	for $\b$-smooth functions, choose $\lambda = \beta/2$, $\g = 1/\b$, and 
	$\mu =0$.
\end{proof}

\subsection{A Lower Bound Example for \texorpdfstring{\MtoM}{Move-towards-Minimizer}}
\label{sec:lbd-m2m}

We show that the \MtoM algorithm is $\Omega(\k)$-competitive, even in
$\R^2$. The essential step of the proof is the following lemma, which
shows that, in a given timestep, $ALG$ can be forced to pay
$\Omega(k)$ times as much as some algorithm $Y=(y_1,\dots, y_t)$ (we
think of $Y$ as a candidate for $OPT$) while at each step $t$, $ALG$ does 
not move any closer to $y_t$.

\begin{lemma} \label{lem:kappa-competitive-step}
Fix $\kappa > 0$. Suppose that $(x_1, \dots, x_{t-1} )$ is 
defined by the \MtoM algorithm and $Y=(y_1,\dots, y_{t-1})$ is a point sequence 
such that $y_{t-1}\neq x_{t-1}$. Define the potential 
\[\Phi_s = \|x_s - y_s\|.\]
Then there is a $\k$-well-conditioned function $f_t$ and a choice of $y_t$ such that
\begin{enumerate}
	\item[(i)] $\cost_t(ALG) \ge \Omega(1)\cdot \Phi_{t-1} \ge \Omega(k) \cdot \cost_t (Y)$ 
	\item[(ii)] $\Phi_t \ge \Phi_{t-1}$, and hence $y_t \neq x_t$.
\end{enumerate} 
\end{lemma}
\begin{proof}
 Observe that if we modify an instance by an isometry the
 algorithm's sequence will also change by the  same isometry. So we
 may assume that 
$x_{t-1} =  (\gamma, \gamma) $ and 
$y_{t-1}= (2\g,0)$, for some $\g > 0$. (See Figure~\ref{fig:mtm-lbd}.) Define
\[f_t(x) = \frac{1}{4 \g}\left(\frac{1}{\k} \cdot x_1^2 + x_2^2\right).\]
Note that $f_t$ is $\k$-well-conditioned. 
It is easily checked that $x_t = \lambda x_{t-1}$ for some $\lambda > \frac{1}{2}$ 
(recall that $x_t$ is chosen to satisfy $f_t(x_t) = \|x_t - x_{t-1}\|$). Thus $ALG$ pays:
\begin{equation}\label{eqn:alg-cost}
\cost_t (ALG) = 2f_t(x_t) = 2\lambda^2 \cdot f_t(x_{t-1}) \ge \Omega(1) \cdot \gamma
	= \Omega(1)\cdot \Phi_{t-1}.
\end{equation}

\begin{figure}[t]
\begin{center}
\includegraphics[height=1.5in]{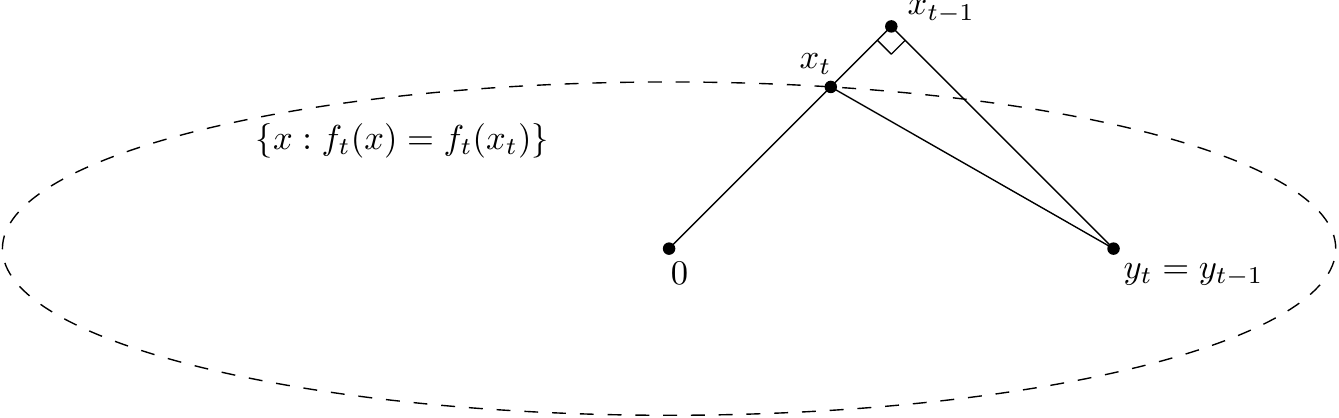}
\end{center}
\caption{Proof of Lemma~\ref{lem:kappa-competitive-step} showing that
  $\MtoM$ is $\Omega(\k)$-competitive.}
\label{fig:mtm-lbd}
\end{figure}

We choose $y_t = y_{t-1}$ so that the cost of $Y$ is:
\begin{equation}\label{eqn:opt-cost}
\cost_t(Y) = f_t(y_t) \le O\left(\frac{1}{\k}\right) \cdot \gamma 
	= O\left(\frac{1}{\k}\right) \cdot \Phi_{t-1}.
\end{equation}
Multiplying (\ref{eqn:opt-cost}) by $\Omega(\k)$ and combining with 
(\ref{eqn:alg-cost}) completes the proof of \emph{(i)}. The statement
in \emph{(ii)} follows from the fact that $x_t, x_{t-1}$ and $y_t$ form a right 
triangle with leg $\Phi_{t-1}$ and hypotenuse $\Phi_t$.
\end{proof}

\begin{proposition}
  \label{prop:m2m-lbd}
The \MtoM algorithm is $\Omega(\k)$ competitive for chasing $\k$-well-conditioned functions.
\end{proposition}
\begin{proof}
Suppose that before the first timestep, $y_0$ moves to $e_1$ and incurs cost 
$1$. Now consider the instance given by repeatedly 
applying Lemma~\ref{lem:kappa-competitive-step} for $T$ timesteps. For each time $t$, we 
have $\Phi_t \ge \Phi_0$. Thus, 
\[\cost_t(ALG) \ge \Omega(1) \cdot \Phi_{t-1} \ge \Omega(1)\cdot \Phi_0 
= \Omega(1).\] 
Summing over all time, $ALG$ pays $\cost(ALG) \ge \Omega(T)$. 
Meanwhile, our candidate $OPT$ has paid at most $O(\frac{1}{\kappa}) \cdot
 \cost(ALG) + 1$. The proof is completed by choosing $T \ge \Omega(\kappa)$.
\end{proof}

\subsection{A Lower Bound Example for \texorpdfstring{\COBD}{COBD}}
\label{sec:lbd-obd}

We now give a lower bound for the \COBD algorithm.\footnote{The
lower bound example is valid even in the unconstrained setting, where \COBD and \OBD are the same algorithm.}
In the proof of Proposition~\ref{thm:sqrt-k-competitive} we showed that the angle $\theta_t$ between $y_t - x_t$ 
and $x_{t-1} - x_t$ satisfies $-\sec(\theta_t) \le O(\sqrt{\k})$. This bound corresponds directly 
determines to the competitiveness of \COBD. The essence of the lower bound is to give an 
example where $-\sec(\theta_t) \ge \Omega(\sqrt{\k})$.

Much like \MtoM, the key to showing that \COBD is $\Omega(\sqrt{\k})$-competitive 
lies in constructing a single ``bad timestep'' that can be repeated until it 
dominates the competitive ratio. In the case of \COBD, this timestep allows us to
convert the potential into cost to $ALG$ at a rate of $\Omega(\sqrt{\k})$.

\begin{lemma}\label{lem:convert-potential}
Fix $\k \ge 1$. Suppose that $x_t$ is defined by the \COBD algorithm and that
$Y=(y_1,\dots, y_{t-1})$ is a point sequence such that $y_{t-1}\neq x_{t-1}$. 
Define the potential 
\[\Phi_s = \|x_s - y_s\|.\]
Then there is a $\k$-well-conditioned function $f_t$ and a choice of $y_t$ such that
\begin{enumerate}
	\item[(i)] $\cost_t(ALG) \ge \Omega(\frac{1}{\sqrt \k}) \Phi_{t-1}$.
        \item[(ii)] $\cost_t(ALG) \ge \Omega(\sqrt{\k}) (-\Delta_t \Phi)$.
	\item[(iii)] $\cost_t(Y) = 0$.
\end{enumerate} 
\end{lemma}
\begin{proof}
Observe that modifying an instance by an isometry will modify the algorithm's sequence 
by the same isometry. After applying an appropriate isometry, we will define 
\[f_t(x) = \a(x_1^2 + \k x_2^2)\]
for some $\a >0$ to be chosen later and $y_t = y_{t-1}$.
We claim that this can be done such that:
\begin{enumerate}
\item[(a)] $y_t = y_{t-1} = 0$,
\item[(b)] $\|x_t - x_{t-1}\| = \frac{1}{2\sqrt \k} \|x_{t-1}-y_{t-1}\|$ (which in turn is equal to   
$\frac{1}{2\sqrt \k} \|x_{t-1}\|$).
\item[(c)] $x_t = \g  \left[\begin{matrix}\sqrt \k \\ 1\end{matrix}\right]$ for some $\gamma>0$,
\end{enumerate}
For any $\alpha > 0$, there is point a $x_\alpha$ on the ray
$\left\lbrace\gamma \left[\begin{matrix}\sqrt \k \\ 1\end{matrix}\right]: 
\gamma >0\right\rbrace$ such that 
$f_t(x_\alpha) = \frac{1}{2\sqrt \k} \|x_{t-1}-y_{t-1}\|$. Let 
\[x^-_\alpha:= x_\alpha + \left(\frac{1}{2\sqrt \k} \|x_{t-1}-y_{t-1}\|\right)
	\frac{\grad f_t(x_\alpha)}{\|\grad f_t(x_\alpha)\|}.\]
Note that $x^-_\alpha$ is defined so that applying \COBD to $x^-_\alpha$ and $f_t$ outputs
the point $x_\alpha$. Then $\|x^-_\alpha\|$ increases continuously from 
$\frac{1}{2\sqrt \k} \|x_{t-1}-y_{t-1}\|$ to $\infty$ as $\alpha$ ranges from $0$ to $\infty$. 
Choose $\a$ such that $\|x^-_\alpha\| =  \|x_{t-1} - y_{t-1}\|$, and pick the isometry that maps 
$y_{t-1}$ to $0$ and $x_{t-1}$ to $x^-_\alpha$. The claim follows.

Now, (a) and (b) imply that
\[\cost_t(ALG) = 2\|x_t-x_{t-1}\| = \frac{1}{\sqrt{\k}}\|x_{t-1}\| = \frac{1}{\sqrt{\k}} \Phi_{t-1}.\]
This proves $(i)$. Furthermore, 
(b) and the triangle inequality give
\begin{equation}
\label{eq:obd-lbd-1}
\|x_t\| 
	\ge \|x_{t-1}\| - \|x_{t-1} - x_t\|
	= (2\sqrt \k- 1) \cdot \|x_{t-1}-x_t\|
	\ge \sqrt\k\cdot  \|x_{t-1}-x_t\|.
\end{equation}

There are $\eta, \nu > 0\footnote{We omit the exact values (which depend on $\k$ and 
$\|x_{t-1}\|$) as $\nu$ cancels out in the next step.}$
such that
\[x_{t-1} - x_t 
	= \eta \grad f_t(x_t) 
	= \nu\left[\begin{matrix}1 \\ \sqrt{\k} \end{matrix}\right].
\]
Letting $\theta_t$ be the angle between $x_{t-1}-x_t$ and $y_t-x_t=-x_t$ (cf.~\Cref{fig:obd-analysis}) we have
\begin{equation} 
\label{eq:obd-lbd-2}
-\cos(\theta_t) 
	= - \frac{\la x_{t-1}-x_t, -x_t \ra}{\|x_{t-1}-x_t\|\cdot \|-x_t\|} 
	= \frac{2\sqrt{\k}}{1+\k} 
	\le \frac{2}{\sqrt{\k}}.
\end{equation}
We now mirror the argument used in the proof of Theorem~\ref{thm:sqrt-k-competitive} relating $\cost_t(ALG)$ to $\cos(\theta_t)$.
\begin{align*}
\cost_t (ALG) &= 2\|x_t-x_{t-1}\|   \\
	&= \frac{\|x_{t-1}\|+\|x_t\| }{\|x_{t-1}-x_t\| - 2\|x_t\|\cos(\theta_t)}
		\cdot (-\Delta_t \Phi) \tag{Law of Cosines, substitution}\\
	&\ge \frac{\|x_t\|}{(1/\sqrt \k) \|x_t\| + (4/\sqrt \k)\cdot \|x_t\|}
		\cdot (-\Delta_t \Phi) \tag{by (\ref{eq:obd-lbd-1}) and (\ref{eq:obd-lbd-2})}\\
	&= \frac{\sqrt\k}{5} (-\Delta_t \Phi).
\end{align*}
Finally, observe that $\cost_t(Y)= f_t(0) = 0$.
\end{proof}

We can now get a lower bound on the competitiveness of \COBD.

\begin{proposition}
  \label{prop:obd-lbd}
\COBD is $\Omega(\sqrt{\k})$ competitive for chasing $\k$-well-conditioned functions.
\end{proposition}
\begin{proof}
Suppose that before the first timestep, $y_0$ moves to $e_1$ and incurs cost 
$1$. Now consider the instance given by repeatedly 
applying Lemma~\ref{lem:convert-potential} for $T$ timesteps. $\cost(OPT)=1$, so it remains to show that 
$\cost(ALG) = \Omega(\sqrt k)$.
Let $\Phi_{min}:= \min\{\Phi_1, \dots, \Phi_T\}$. 
Using $(i)$ and summing over all time we have 
\begin{equation}\label{eq:obd-lbd-3}
\cost(ALG) \ge \frac{1}{\sqrt{\k}} \sum_{t=0}^{T-1} \Phi_t \ge \frac{T}{\sqrt{\k}} \Phi_{min}.
\end{equation}
Using $(ii)$ and summing over all time (and using that $ALG$ incurs nonnegative cost at each step),
\begin{equation}\label{eq:obd-lbd-4}
\cost(ALG) \ge \Omega(\sqrt \k)(\Phi_0 - \Phi_{min}) =  \Omega(\sqrt \k)(1- \Phi_{min})
\end{equation}
If $\Phi_{min} \ge \frac{1}{2}$ then $\cost(ALG) \ge \frac{T}{2\sqrt k}$ by 
(\ref{eq:obd-lbd-3}), else  $\Phi_{min} < \frac{1}{2}$ and we have 
$\cost(ALG) \ge \Omega(\sqrt k)$ by (\ref{eq:obd-lbd-4}).
Choosing $T = \kappa$ completes the proof.
\end{proof}

%% file: general-norm.tex
\section{Constrained \MtoM}
\label{sec:constrained-m2m}

We give a generalized version of the \MtoM algorithm for the constrained setting where the action space $K\sse \R^d$ is an arbitrary convex set. This algorithm achieves the same $O(\sqrt{\kappa})$-competitiveness respectively as in the unconstrained setting.

The idea is to move towards $x_{K,t}^*$, the minimizer of $f_t$ among feasible points, rather than the global minimizer. The proof of the algorithm's competitiveness proceeds similarly to the proof in the unconstrained setting. The difference is that it takes more care to show that $f(x_t) \le O(\k) f(y_t)$ in Case II.

\begin{quote}
  \textbf{The Constrained M2M Algorithm.} Suppose we are at position $x_{t-1}$ and receive the
  function $f_t$. Let $x^*_{K,t} := \arg\min_{x\in K} f_t(x)$ denote the
  minimizer of $f_t$ among points in $K$. Consider the line segment with endpoints
  $x_{t-1}$ and $x^*_{K,t}$, and let $x_t$ be the unique point on this
  segment with $\|x_t - x_{t-1}\| = f_t(x_t)-f_t(x_{K,t}^*)$.\footnote{Such a 
  point is always unique when $f_t$ is strictly convex.} 
  The point $x_t$ is the one played by the algorithm.
\end{quote}

Note that we assume that the global minimum value of $f_t$ is $0$, as before. However, the minimum value of $f_t$ on the action space $K$ could be strictly positive.

\begin{proposition}\label{prop: mtm-constrained}
  With $c = 25(2+2\sqrt 2)$, for each $t$,
  \begin{gather}
    \cost_t(ALG) + 2\sqrt 2 \cdot \Delta_t \Phi \le c\cdot \k \cdot
    \cost_t(OPT). \label{eq:c-1}
  \end{gather}
  Hence, the constrained \MtoM algorithm is $c \k$-competitive.
\end{proposition}
\begin{proof}
As in the proof of Theorem~\ref{thm:k-competitive}, we begin by applying the structure lemma. This time, we use $x^*_{K,t}$ to be the origin. The proof of Case I is identical. 

\medskip\textbf{Case II:} Suppose that $\|y_t - x^*_{K,t}\| \ge \frac{1}{\sqrt{2}} \|x_t - x^*_{K,t}\|$. Let $x^*_t := \arg\min_{x} f_t(x)$ denote the global minimizer of $f_t$. As before, we assume $f_t(x^*_t) = 0$, and we translate such that $x^*_t = 0$. 

We now show that $f_t(x_t) \le 25\k f_t(y_t)$. 
If $f_t(x_t) \le 25 \k f_t(x^*_{K,t})$, then since $f(y_t) \ge f_t(x^*_{K,t})$, we are done. So suppose that $f_t(x_t) > 25 \k f_t(x^*_{K,t})$. Now strong convexity and smoothness imply
\begin{equation}
\|x_t\|^2 \ge \frac{2}{\k \a_t} f_t(x_t) \ge 25\cdot \frac{2}{\a_t} f(x^*_{K,t}) \ge 25 \|x^*_{K,t}\|^2.
\end{equation}
Thus $\|x_t\| \ge 5 \|x^*_{K,t}\|$. 
One application of the triangle inequality gives
$\|x_t - x^*_{K,t}\| \ge \|x_t\| - \|x^*_{K,t}\| \ge 4\|x^*_{K,t}\|$.
Using the triangle inequality again, we get 
\begin{equation}
\|x_t\| \le \|x_t-x^*_{K,t}\| + \|x^*_{K,t}\| \le \frac{5}{4} \|x-x^*_{K,t}\|,
\end{equation}
and
\begin{equation}
\|y_t\| 
	\ge \|y-x^*_{K,t}\| - \|x^*_{K,t}\| 
	\ge \left(\frac{1}{\sqrt2} - \frac{1}{4}\right) \|x-x^*_{K,t}\| 
	\ge \frac{1}{4} \|x-x^*_{K,t}\| 
\end{equation}
Combining these two, we have
\begin{equation}
\|y_t\| \ge \|x_t\| \ge \frac{1}{5}\|x_t\|
\end{equation}
Finally, we have 
\begin{equation}\label{eq:general-norm-1}
	f_t(x_t)
		\le \frac{\alpha_t \kappa}{2} \|x_t\|^2
		\le \frac{5\alpha_t\k}{2} \|y_t\|^2
		\le 25 \kappa \cdot f_t(y_t).
\end{equation}
We now proceed as in the proof of Theorem~\ref{thm:k-competitive}. 
\end{proof}

\section{A Structure Lemma for General Norms}
\label{sec:general-norm}

We can extend the $O(\k)$-competitiveness guarantee for \MtoM for all norms, by 
replacing Lemma~\ref{lem:structure} by the following Lemma~\ref{lem:general-norm} 
in \Cref{thm:k-competitive}, and changing some of the constants in the latter 
accordingly.

\begin{lemma}
  \label{lem:general-norm}
  Fix an arbitrary norm $\|\cdot \|$ on $\R^d$.
  Given any scalar $\g\in [0,1]$ and any two vectors $x,y\in \R^d$, at least one of the following holds:
  \begin{enumerate}
  \item[(i)] $ \|y-\g x\| - \|y-x\| \le -\frac{1}{2} \|x-\g x\|$.
  \item[(ii)] $\|y\| \ge \frac{1}{4}\|\g x\|$.
  \end{enumerate}
\end{lemma}

\begin{proof}
As in the proof of Lemma~\ref{lem:structure} we assume (ii) does not hold and 
  show that (i) does. WLOG, let $\|x\|=1$. Let $\|\cdot\|_*$ denote the dual norm. 
  Let $z_\tau := \grad \|\tau x - y\|= \arg \max_{\|z\|_*\le 1} 
  \langle \tau x - y, z\rangle$ and note that $\langle z_\tau, \tau x - y\rangle =
    \|\tau x -y \|$. Then, 
  \begin{align*}
  	\frac{d}{d\tau} \|\tau x - y\|
  		&= \left\langle \grad\|\tau x - y\|,\frac{d}{d\tau}(\tau x-y)\right\rangle\\
  		&= \left\langle z_\tau , x\right\rangle\\
   		&= \frac{\langle z_\tau, \tau x - y \rangle
   			+ \langle z_\tau, y \rangle }{\tau} \\
   		&\ge \frac{\|\tau x - y\| - \|z_\tau\|_* \|y\|}{\tau}
   			\tag{definition of $z_\tau$ and H\"older}\\
   		&\ge \frac{(\tau - \|y\|) - 1\cdot \|y\|}{\tau}
   		= 1 - \frac{2\|y\|}{\tau}.
   			\tag{triangle inequality}
  \end{align*}
Given the bound $\frac{d}{d\tau} \|\tau x - y\| \ge 1 - \frac{2\|y\|}{\tau}$ we can say:
\begin{gather}
  \|y - \gamma x\| - \|y-x\| = - \int_\gamma^1 \frac{d}{d\tau}
  \big(\|\tau x - y\|\big) \, d\tau \le - \int_\gamma^1 \left( 1 -
    \frac{2\|y\|}{\tau} \right)\, d\tau. \label{eq:5}
\end{gather}
Since by assumption condition~(ii) does hold and $\|x\| = 1$, we know that $\|y\| <
\frac14 \|\g x\| = \frac14 \g$. Hence $\frac{2\|y\|}{\tau} <
\frac{\g/2}{\tau} \leq 1/2$ for $\tau \geq \g$. The integrand in~(\ref{eq:5}) is
therefore at least half, and hence the result is at most $-\frac{1}{2}(1-\gamma)
	= -\frac{1}{2}\|x - \gamma x\|$. Hence the proof.
\end{proof}